\documentclass[letterpaper,11pt]{article}
\usepackage{amssymb}
\usepackage{amsmath}
\usepackage{pslatex}
\usepackage{fullpage}

\usepackage[ruled,vlined,linesnumbered]{algorithm2e}
\dontprintsemicolon

\DeclareMathSymbol{\qedsymb} {\mathord}{AMSa}{"04}
\newcommand{\qed}{\hspace*{0pt}\hfill$\qedsymb$}
\newcommand{\eps}{\varepsilon}

\newtheorem{definition}{Definition}
\newtheorem{lemma}[definition]{Lemma}

\newenvironment{proof}{\noindent{\bf Proof}\ \ }{\qed\medskip}
\usepackage{times}

\DeclareMathOperator{\poly}{poly}

\begin{document}

\title{Testing Distribution Identity Efficiently}
\author{Krzysztof Onak\\ MIT\\ {\tt konak@mit.edu}\iffalse \and Ronitt Rubinfeld\\ MIT, Tel Aviv University\\ {\tt ronitt@csail.mit.edu}\fi}
\date{}

\maketitle

\begin{abstract}
We consider the problem of testing distribution identity. Given a sequence of independent samples from an unknown distribution on a domain of size $n$, the goal is to check if the unknown distribution approximately equals a known distribution on the same domain. While Batu, Fortnow, Fischer, Kumar, Rubinfeld, and White (FOCS 2001) proved that the sample complexity of the problem is $\tilde O(\sqrt{n} \cdot \poly(1/\eps))$, the running time of their tester is much higher: $O(n) + \tilde O(\sqrt{n} \cdot \poly(1/\eps))$.
We modify their tester to achieve a running time of $\tilde O(\sqrt{n}\cdot \poly(1/\eps))$.
\end{abstract}

Let $p$ and $q$ be two probability distributions on $[n]$\footnote{We write $[k]$ to denote the set $\{1,2,\ldots,k\}$, for any positive integer $k$.}, and let $\|p-q\|_1$ denote the $\ell_1$-distance between $p$ and $q$. In this paper, algorithms have access to two distributions $q$ and $p$.
\begin{itemize}
 \item The distribution $p$ is \emph{known}: for each $i \in [n]$, the algorithm can query the probability $p_i$ of $i$ in constant time.
 \item The distribution $q$ is \emph{unknown}: the algorithm can only obtain an independent \emph{sample} from $q$ in constant time.
\end{itemize}
An \emph{identity tester} is an algorithm such that:
\begin{itemize}
\item if $p=q$, then it accepts with probability $2/3$,
\item if $\|p-q\|_1 \ge \eps$, then it rejects with probability $2/3$.
\end{itemize}

Batu, Fortnow, Fischer, Kumar, Rubinfeld, and White \cite{BFFKRW} proved that there is an identity tester that uses only $\tilde O(\sqrt{n} \cdot \poly(1/\eps))$ samples from $q$. A shortcoming of their algorithm is a running time of $O(n) + \tilde O(\sqrt{n} \cdot \poly(1/\eps))$. In this note, we show that their tester can be modified to achieve a running time of $\tilde O(\sqrt{n} \cdot \poly(1/\eps))$. It is also well known that $\Omega(\sqrt{n})$ samples are required to tell the uniform distribution on $[n]$ from a distribution that is uniform on a random subset of $[n]$ of size $n/2$.

\section{The Original Tester}

We now describe the tester of Batu\ {\it et al.}\ \cite{BFFKRW}, which is outlined as Algorithm~\ref{alg:batu}.
Let $\eps'=\eps/C$, where $C$ is a sufficiently large positive constant.
The tester starts by partitioning the set $[n]$ into $k+1 =\left\lceil\log_{1+\eps'}\frac{2n}{\eps}\right\rceil + 1 = O(\frac{1}{\eps}\cdot\log (n/\eps))$ sets $R_0$, $R_1$, \ldots, $R_k$ in Step~1, where 
$$R_j = \left\{i \in [n]: \frac{\eps}{2n}\cdot (1+\eps')^{j-1} < p_i \le \frac{\eps}{2n}\cdot (1+\eps')^{j} \right\}$$
for $j > 0$, and
$$R_0 = \left\{i \in [n]: p_i \le \frac{\eps}{2n} \right\}.$$
We then define probabilities of each set according to $p$ and $q$: $P_j=\sum_{i \in R_j} p_i$ and $Q_j=\sum_{i \in R_j} q_i$. The tester computes and estimates those probabilities in Steps~2 and~3.
In Step~4, the tester verifies that the probabilities of sets $R_j$ in both the distributions are close. Finally, in Steps~5--7, the tester verifies that $q$ restricted to each $R_j$ is approximately uniform, by comparing second moments of $p$ and $q$ over each $R_j$.
If $q$ passes the test with probability greater than $1/3$, it must be close to $p$. On the other hand, if $p=q$, then the parameters can be set so that $q$ passes with probability $2/3$.

\begin{algorithm}[t]
\label{alg:batu}
\caption{Outline of the tester of Batu {\it et al.}\ \cite{BFFKRW}}

Partition $[n]$ into $R_0$, $R_1$, \ldots, $R_k$ \;
Compute $P_j$, for $j \in \{0,1,\ldots,k\}$\;
Use $O((k/\eps)^2\cdot\log k)$ samples from $q$ to get an estimate $Q'_j$ of each $Q_j$ up to $\eps/(4k+4)$\;
\lIf{$\|(P_0,\ldots,P_k)-(Q'_0,\ldots,Q'_k)\|_1 > \eps/4$}{{\bf REJECT}}\;
Let $s_i$, $i\in[n]$, be the number of occurrences of $i$ in a sample of size $S = \tilde O(\sqrt{n} \cdot \poly(1/\eps))$\;
\For{$j > 0$ s.t.\ $P_j > \eps/(4k+4)$}{
\lIf{$\sum_{i\in R_j} \binom{s_i}{2} > (1+\eps/4) \cdot \binom{Q}{2} \cdot P_j \cdot \frac{\eps}{2n}(1+\eps')^j$}{\bf REJECT}
}
{\bf ACCEPT}
\end{algorithm}

Note that the additive linear term in the complexity of the tester comes from explicitly computing each $R_i$ and each $P_i$ in Steps~1--2.

\section{Our Improvement}

Note that the partition of $[n]$ into sets $R_j$ need not be computed explicitly, since
for each sample $i$ from $q$, one can check which $R_j$ it belongs to by querying $p_i$.

We observe that one can verify that $\|(P_0,\ldots,P_k)-(Q_0,\ldots,Q_k)\|_1$ is small without explicitly computing each $P_i$. We use Algorithm~\ref{alg:ineq} for this purpose. Let $j^\star$ be an index such that an element of probability $1/\sqrt{n}$ would belong to $R_{j^\star}$. The algorithm is based on the following facts:
\begin{itemize}
 \item For $j < j^\star$, if $P_j$ is not negligible, $R_j$ must be large, and a good additive estimate to $P_j$
can be obtained by uniformly sampling $\tilde O(\sqrt{n} \cdot \poly(1/\eps))$ elements of $[n]$, and computing the weight of those that belong to $R_j$.

 \item If $p=q$, we are likely to learn all elements in $R_j$, $j \ge j^\star$, by sampling only $\tilde O(\sqrt{n})$ elements of $q$. This gives the exact value of each $P_j$, $j > j^\star$. If $p\ne q$, this method still gives lower bounds for each $P_j$.
\end{itemize}
If $\|(P_0,\ldots,P_k)-(Q'_0,\ldots,Q'_k)\|_1 \ge \delta$, our estimates for $P_j$ and $Q_j$ are likely to be sufficiently different. A detailed proof follows.

\begin{algorithm}[t]
\label{alg:ineq}
\caption{Telling $p=q$ (Case 1) from $\|(P_0,\ldots,P_k)-(Q_0,\ldots,Q_k)\|_1 \ge \delta$ (Case 2)}
Use $O((k/\delta)^2\cdot\log k)$ samples from $q$ to get an estimate $Q'_j$ of each $Q_j$ up to $\delta/(8k+8)$ \;
Let $j^\star$ be an index such that an element of probability $1/\sqrt{n}$ would belong to $R_{j^\star}$ \;
Let $S_1$ be a set of $O(\sqrt{n} \cdot \log n)$ samples from $q$\;
\For{$j$ s.t.\ $j^\star \le j \le k$}{
Let $T_{j} = S_1 \cap R_j$ (with no repetitions)\;
\lIf{$ \left|Q'_j - \sum_{i \in T_j} p_i \right| >  \frac{\delta}{8k+8}$}{{\bf return} ``Case 2''}
}
Let $S_2$ be a set of $O\left(\left(\frac{k}{\delta}\right)^3\sqrt{n}\cdot\log k\right)$ independent samples from $[n]$ with replacement\;
\For{$j$ s.t.\ $j < j^\star$}{
Let $U_j = S_2 \cap R_j$ (with repetitions)\;
\lIf{$\left| Q'_j - \frac{\sum_{i\in U_j} p_i}{|S_2|} \right| >  \frac{\delta}{4k+4}$}{{\bf return} ``Case 2''}
}
{{\bf return} ``Case 1''}
\end{algorithm}

\begin{lemma}
Algorithm~\ref{alg:ineq} with appropriately chosen constants tells $p=q$ (Case~1) from
$\|(P_0,\ldots,P_k)-(Q_0,\ldots,Q_k)\|_1 \ge \delta$ (Case~2) with probability $9/10$.
\end{lemma}

\begin{proof}
The multiplicative constant in the sample size of Step~1 is such that Step~1 succeeds with probability $99/100$. The size of $S_1$ is chosen such that with probability $99/100$, $S_1$ contains all elements $i$ of probability $q_i \ge \frac{1}{2\sqrt{n}}$ by the coupons collector's problem. Finally, the size of $S_2$ is chosen such that with probability $99/100$, for each $j < j^\star$, $\left|\frac{\sum_{i\in U_j} p_i}{|S_2|} - P_j\right| \le \frac{\delta}{8k+8}$. To see this, let us first focus on $j < j^\star$ such that $P_j \ge \frac{\delta}{16(k+1)}$. Note that each $i \in S_2$ contributes with a value in $[0,1/\sqrt{n}]$ to $\sum_{i\in U_j} p_i$. By the Chernoff bound, $O\left(\left(\frac{k}{\delta}\right)^3\sqrt{n}\cdot\log k\right)$ samples suffice to estimate $P_j$ with multiplicative error $1+\frac{\delta}{8k+8}$ with probability $1 - \frac{1}{200k}$, which implies additive error at most $\frac{\delta}{8k+8}$ as well. For $j < j^\star$ such that $P_j < \frac{\delta}{16(k+1)}$, the Chernoff bound still guarantees with the same probability that the estimate is less than $\frac{\delta}{8k+8}$.

If $p=q$, then Algorithm~\ref{alg:ineq} discovers this with probability $97/100$ due to the following facts. Firstly, $T_j = R_j$, for $j \ge j^\star$, so $\sum_{i \in T_j} p_i = P_j$. Therefore, provided all $Q'_j$ are good approximations to the corresponding $Q_j$, $q$ always passes Step~6. Secondly, if all $\frac{\sum_{i\in U_j} p_i}{|S_2|}$, $0 \le j  < j^\star$, are good approximations of the corresponding $P_j$, $q$ always passes Step~10 as well.

If $\|(P_0,\ldots,P_k)-(Q_0,\ldots,Q_k)\|_1 \ge \delta$, there is $j'$ such that $Q_{j'} > P_{j'} + \frac{\delta}{2k}$. If $j' \ge j^\star$, then because $P_j$ is always greater than or equal to $\sum_{i \in T_j} p_i$, the tester concludes in Step~6 for $j = j'$ that Case~2 occurs, provided $Q'_{j'}$ is a good approximation to $Q_{j'}$, which happens with probability at least $99/100$. If $j' < j^\star$, then because we have good approximations to both $Q_{j'}$ and $P_{j'}$ with probability $98/100$, and their distance is at least $\frac{\delta}{2k} - 2 \frac{\delta}{8k+8} > \frac{\delta}{4k+4}$, the algorithm concludes in Step~10 for $j = j'$ that Case~2 occurs.
\end{proof}

To get an efficient tester, we replace Steps~2--4 of Algorithm~\ref{alg:batu} with Algorithm~\ref{alg:ineq}, where we set $\delta$ to $\eps/C$ for a sufficiently large constant $C$. If Algorithm~\ref{alg:ineq} concludes that Case~2 occurs, the new algorithm immediately rejects. Furthermore, if it is not the case that $\|(P_0,\ldots,P_k)-(Q_0,\ldots,Q_k)\|_1 \ge \delta$, Steps~6 and~7 work with estimates $Q'_j$ instead of the exact values $P_j$ up to a modification of constants.

\subsection*{Acknowledgment}

The author thanks Ronitt Rubinfeld for asking the question.

\bibliographystyle{alpha}
\bibliography{efficient_identity}

\newcommand{\etalchar}[1]{$^{#1}$}
\begin{thebibliography}{BFF{\etalchar{+}}01}

\bibitem[BFF{\etalchar{+}}01]{BFFKRW}
Tugkan Batu, Lance Fortnow, Eldar Fischer, Ravi Kumar, Ronitt Rubinfeld, and
  Patrick White.
\newblock Testing random variables for independence and identity.
\newblock In {\em FOCS}, pages 442--451, 2001.

\end{thebibliography}

\end{document}